\documentclass{amsart}
\usepackage{amsfonts}
\usepackage{amsmath,amscd}
\usepackage{amsthm}
\usepackage{amssymb}
\usepackage{latexsym}
\newtheorem{cor}{Corollary}[section]

\newtheorem{prop}{Proposition}[section]

\newtheorem{example}{Example}
\newtheorem{defn}{Definition}[section]
\newtheorem{thm}{Theorem}

\newtheorem{rem}{Remark}
\newtheorem{notation}{Notation}

\numberwithin{equation}{section}
\begin{document}
\newcommand{\beqa}{\begin{eqnarray}}
\newcommand{\eeqa}{\end{eqnarray}}
\newcommand{\thmref}[1]{Theorem~\ref{#1}}
\newcommand{\secref}[1]{Sect.~\ref{#1}}
\newcommand{\lemref}[1]{Lemma~\ref{#1}}
\newcommand{\propref}[1]{Proposition~\ref{#1}}
\newcommand{\corref}[1]{Corollary~\ref{#1}}
\newcommand{\remref}[1]{Remark~\ref{#1}}
\newcommand{\er}[1]{(\ref{#1})}
\newcommand{\nc}{\newcommand}
\newcommand{\rnc}{\renewcommand}

\nc{\cal}{\mathcal}

\nc{\goth}{\mathfrak}
\rnc{\bold}{\mathbf}
\renewcommand{\frak}{\mathfrak}
\renewcommand{\Bbb}{\mathbb}

\def\cN{{\cal N}}
\newcommand{\hs}[1]{\hspace{#1 mm}}
\newcommand{\mb}[1]{\hs{4}\mbox{#1}\hs{4}}
\newcommand{\id}{\text{id}}
\nc{\Cal}{\mathcal}
\nc{\Xp}[1]{X^+(#1)}
\nc{\Xm}[1]{X^-(#1)}
\nc{\on}{\operatorname}
\nc{\ch}{\mbox{ch}}
\nc{\Z}{{\bold Z}}
\nc{\J}{{\mathcal J}}
\nc{\C}{{\bold C}}
\nc{\Q}{{\bold Q}}
\renewcommand{\P}{{\mathcal P}}
\nc{\N}{{\Bbb N}}
\nc\beq{\begin{equation}}
\nc\enq{\end{equation}}
\nc\lan{\langle}
\nc\ran{\rangle}
\nc\bsl{\backslash}
\nc\mto{\mapsto}
\nc\lra{\leftrightarrow}
\nc\hra{\hookrightarrow}
\nc\sm{\smallmatrix}
\nc\esm{\endsmallmatrix}
\nc\sub{\subset}
\nc\ti{\tilde}
\nc\nl{\newline}
\nc\fra{\frac}
\nc\und{\underline}
\nc\ov{\overline}
\nc\ot{\otimes}
\nc\bbq{\bar{\bq}_l}
\nc\bcc{\thickfracwithdelims[]\thickness0}
\nc\ad{\text{\rm ad}}
\nc\Ad{\text{\rm Ad}}
\nc\Hom{\text{\rm Hom}}
\nc\End{\text{\rm End}}
\nc\Ind{\text{\rm Ind}}
\nc\Res{\text{\rm Res}}
\nc\Ker{\text{\rm Ker}}
\rnc\Im{\text{Im}}
\nc\sgn{\text{\rm sgn}}
\nc\tr{\text{\rm tr}}
\nc\Tr{\text{\rm Tr}}
\nc\supp{\text{\rm supp}}
\nc\card{\text{\rm card}}
\nc\bst{{}^\bigstar\!}
\nc\he{\heartsuit}
\nc\clu{\clubsuit}
\nc\spa{\spadesuit}
\nc\di{\diamond}
\nc\cW{\cal W}
\nc\cG{\cal G}
\nc\al{\alpha}
\nc\bet{\beta}
\nc\ga{\gamma}
\nc\de{\delta}
\nc\ep{\epsilon}
\nc\io{\iota}
\nc\om{\omega}
\nc\si{\sigma}
\rnc\th{\theta}
\nc\ka{\kappa}
\nc\la{\lambda}
\nc\ze{\zeta}

\nc\vp{\varpi}
\nc\vt{\vartheta}
\nc\vr{\varrho}

\nc\Ga{\Gamma}
\nc\De{\Delta}
\nc\Om{\Omega}
\nc\Si{\Sigma}
\nc\Th{\Theta}
\nc\La{\Lambda}

\nc\boa{\bold a}
\nc\bob{\bold b}
\nc\boc{\bold c}
\nc\bod{\bold d}
\nc\boe{\bold e}
\nc\bof{\bold f}
\nc\bog{\bold g}
\nc\boh{\bold h}
\nc\boi{\bold i}
\nc\boj{\bold j}
\nc\bok{\bold k}
\nc\bol{\bold l}
\nc\bom{\bold m}
\nc\bon{\bold n}
\nc\boo{\bold o}
\nc\bop{\bold p}
\nc\boq{\bold q}
\nc\bor{\bold r}
\nc\bos{\bold s}
\nc\bou{\bold u}
\nc\bov{\bold v}
\nc\bow{\bold w}
\nc\boz{\bold z}

\nc\ba{\bold A}
\nc\bb{\bold B}
\nc\bc{\bold C}
\nc\bd{\bold D}
\nc\be{\bold E}
\nc\bg{\bold G}
\nc\bh{\bold H}
\nc\bi{\bold I}
\nc\bj{\bold J}
\nc\bk{\bold K}
\nc\bl{\bold L}
\nc\bm{\bold M}
\nc\bn{\bold N}
\nc\bo{\bold O}
\nc\bp{\bold P}
\nc\bq{\bold Q}
\nc\br{\bold R}
\nc\bs{\bold S}
\nc\bt{\bold T}
\nc\bu{\bold U}
\nc\bv{\bold V}
\nc\bw{\bold W}
\nc\bz{\bold Z}
\nc\bx{\bold X}

\nc\ca{\mathcal A}
\nc\cb{\mathcal B}
\nc\cc{\mathcal C}
\nc\cd{\mathcal D}
\nc\ce{\mathcal E}
\nc\cf{\mathcal F}
\nc\cg{\mathcal G}
\rnc\ch{\mathcal H}
\nc\ci{\mathcal I}
\nc\cj{\mathcal J}
\nc\ck{\mathcal K}
\nc\cl{\mathcal L}
\nc\cm{\mathcal M}
\nc\cn{\mathcal N}
\nc\co{\mathcal O}
\nc\cp{\mathcal P}
\nc\cq{\mathcal Q}
\nc\car{\mathcal R}
\nc\cs{\mathcal S}
\nc\ct{\mathcal T}
\nc\cu{\mathcal U}
\nc\cv{\mathcal V}
\nc\cz{\mathcal Z}
\nc\cx{\mathcal X}
\nc\cy{\mathcal Y}

\nc\e[1]{E_{#1}}
\nc\ei[1]{E_{\delta - \alpha_{#1}}}
\nc\esi[1]{E_{s \delta - \alpha_{#1}}}
\nc\eri[1]{E_{r \delta - \alpha_{#1}}}
\nc\ed[2][]{E_{#1 \delta,#2}}
\nc\ekd[1]{E_{k \delta,#1}}
\nc\emd[1]{E_{m \delta,#1}}
\nc\erd[1]{E_{r \delta,#1}}

\nc\ef[1]{F_{#1}}
\nc\efi[1]{F_{\delta - \alpha_{#1}}}
\nc\efsi[1]{F_{s \delta - \alpha_{#1}}}
\nc\efri[1]{F_{r \delta - \alpha_{#1}}}
\nc\efd[2][]{F_{#1 \delta,#2}}
\nc\efkd[1]{F_{k \delta,#1}}
\nc\efmd[1]{F_{m \delta,#1}}
\nc\efrd[1]{F_{r \delta,#1}}

\nc\fa{\frak a}
\nc\fb{\frak b}
\nc\fc{\frak c}
\nc\fd{\frak d}
\nc\fe{\frak e}
\nc\ff{\frak f}
\nc\fg{\frak g}
\nc\fh{\frak h}
\nc\fj{\frak j}
\nc\fk{\frak k}
\nc\fl{\frak l}
\nc\fm{\frak m}
\nc\fn{\frak n}
\nc\fo{\frak o}
\nc\fp{\frak p}
\nc\fq{\frak q}
\nc\fr{\frak r}
\nc\fs{\frak s}
\nc\ft{\frak t}
\nc\fu{\frak u}
\nc\fv{\frak v}
\nc\fz{\frak z}
\nc\fx{\frak x}
\nc\fy{\frak y}

\nc\fA{\frak A}
\nc\fB{\frak B}
\nc\fC{\frak C}
\nc\fD{\frak D}
\nc\fE{\frak E}
\nc\fF{\frak F}
\nc\fG{\frak G}
\nc\fH{\frak H}
\nc\fJ{\frak J}
\nc\fK{\frak K}
\nc\fL{\frak L}
\nc\fM{\frak M}
\nc\fN{\frak N}
\nc\fO{\frak O}
\nc\fP{\frak P}
\nc\fQ{\frak Q}
\nc\fR{\frak R}
\nc\fS{\frak S}
\nc\fT{\frak T}
\nc\fU{\frak U}
\nc\fV{\frak V}
\nc\fZ{\frak Z}
\nc\fX{\frak X}
\nc\fY{\frak Y}
\nc\tfi{\ti{\Phi}}
\nc\bF{\bold F}
\rnc\bol{\bold 1}

\nc\ua{\bold U_\A}

\def\cA{{\cal A}}   \def\cB{{\cal B}}   \def\cC{{\cal C}}
\def\cD{{\cal D}}   \def\cE{{\cal E}}   \def\cF{{\cal F}}  
\def\cG{{\cal G}}   \def\cH{{\cal H}}   \def\cI{{\cal I}}
\def\cJ{{\cal J}}   \def\cK{{\cal K}}   \def\cL{{\cal L}}
\def\cM{{\cal M}}   \def\cN{{\cal N}}   \def\cO{{\cal O}}
\def\cP{{\cal P}}   \def\cQ{{\cal Q}}   \def\cR{{\cal R}}
\def\cS{{\cal S}}   \def\cT{{\cal T}}   \def\cU{{\cal U}}
\def\cV{{\cal V}}   \def\cW{{\cal W}}   \def\cX{{\cal X}}
\def\cY{{\cal Y}}   \def\cZ{{\cal Z}}  \def\cRR{{\cal {\mathbb R}}}

\newcommand{\ben}{\begin{eqnarray}}
\newcommand{\een}{\end{eqnarray}}
\newcommand{\nonu}{\nonumber \\} 
\newcommand{\FF}{{\mathbb F}}
\newcommand{\A}{{\mathbb A}}
\newcommand{\GG}{{\mathbb G}}

\nc\qinti[1]{[#1]_i}
\nc\q[1]{[#1]_q}
\nc\xpm[2]{E_{#2 \delta \pm \alpha_#1}}  
\nc\xmp[2]{E_{#2 \delta \mp \alpha_#1}}
\nc\xp[2]{E_{#2 \delta + \alpha_{#1}}}
\nc\xm[2]{E_{#2 \delta - \alpha_{#1}}}
\nc\hik{\ed{k}{i}}
\nc\hjl{\ed{l}{j}}
\nc\qcoeff[3]{\left[ \begin{smallmatrix} {#1}& \\ {#2}& \end{smallmatrix}
\negthickspace \right]_{#3}}
\nc\qi{q}
\nc\qj{q}

\nc\ufdm{{_\ca\bu}_{\rm fd}^{\le 0}}


\nc\isom{\cong} 

\nc{\pone}{{\Bbb C}{\Bbb P}^1}
\nc{\pa}{\partial}
\def\H{\mathcal H}
\def\L{\mathcal L}
\nc{\F}{{\mathcal F}}
\nc{\Sym}{{\goth S}}
\nc{\arr}{\rightarrow}
\nc{\larr}{\longrightarrow}

\nc{\ri}{\rangle}
\nc{\lef}{\langle}
\nc{\W}{{\mathcal W}}
\nc{\uqatwoatone}{{U_{q,1}}(\su)}
\nc{\uqtwo}{U_q(\goth{sl}_2)}
\nc{\dij}{\delta_{ij}}
\nc{\divei}{E_{\alpha_i}^{(n)}}
\nc{\divfi}{F_{\alpha_i}^{(n)}}
\nc{\Lzero}{\Lambda_0}
\nc{\Lone}{\Lambda_1}
\nc{\ve}{\varepsilon}
\nc{\phioneminusi}{\Phi^{(1-i,i)}}
\nc{\phioneminusistar}{\Phi^{* (1-i,i)}}
\nc{\phii}{\Phi^{(i,1-i)}}
\nc{\Li}{\Lambda_i}
\nc{\Loneminusi}{\Lambda_{1-i}}
\nc{\vtimesz}{v_\ve \otimes z^m}

\nc{\asltwo}{\widehat{\goth{sl}_2}}
\nc\ag{\widehat{\goth{g}}}  
\nc\teb{\tilde E_\boc}
\nc\tebp{\tilde E_{\boc'}}

\title[A note on the $O_q(\widehat{sl_2})$ algebra]{A note on the $O_q(\widehat{sl_2})$ algebra}
\author{P. Baseilhac}
\address{Laboratoire de Math\'ematiques et Physique Th\'eorique CNRS/UMR 6083,
          F\'ed\'eration Denis Poisson, Universit\'e de Tours, Parc de Grammont, 37200 Tours, FRANCE}
\email{baseilha@lmpt.univ-tours.fr}

\author{S. Belliard}
\address{Istituto Nazionale di Fisica Nucleare, Sezione di Bologna, Via Irnerio 46,  40126 Bologna,  Italy}
\email{belliard@bo.infn.it}

\begin{abstract}
An explicit homomorphism that relates the elements of the infinite dimensional non-Abelian algebra generating $O_q(\widehat{sl_2})$ currents and the standard generators of the $q-$Onsager algebra is proposed. Two straightforward applications of the result are then considered: First, for the class of quantum integrable models which integrability condition originates in the $q-$Onsager spectrum generating algebra, the infinite $q-$deformed Dolan-Grady hierarchy is derived - bypassing the transfer matrix formalism. Secondly, higher Askey-Wilson relations that arise in the study of symmetric special functions generalizing the Askey-Wilson $q-$orthogonal polynomials are proposed.
\end{abstract}

\maketitle

\vskip -0.5cm


\vspace{0cm}

\section{Introduction}
Current algebras are known to play an important role in the theory of quantum integrable systems and, more generally, in the study of mathematical structures such as quantum groups. In particular, they have found useful and numerous applications in the derivation of correlation functions of physical observables in conformal field theory \cite{KZ} or quantum integrable spin chain \cite{Jimbo1,vertex}. Among these, Drinfeld current algebras and their representation theory have attracted much attention. Their simplest representative associated with the $U_q(\widehat{sl_2})$ algebra is, for instance, at the root of the vertex operators' solution to the thermodynamic limit of the XXZ spin chain. Indeed, in this non-perturbative approach the isomorphism between the Yang-Baxter's algebraic structure \cite{FRT1}, the first Drinfeld-Jimbo presentation (Chevalley basis) \cite{D1,J1} and second  Drinfeld's presentation (current algebra) \cite{D2} of $U_q(\widehat{sl_2})$ plays a crucial role. Note that explicit homomorphisms between these presentations were proposed in \cite{RS,DF} and \cite{Tol,Lev,Dam} (see also \cite{Be}) based on Luztig's automorphism \cite{L} (see also \cite{Jin}). \vspace{1mm} 

Another type of current algebra has recently been proposed in connection with the reflection equation algebra. Similarly to Drinfeld current algebras, it possesses promising applications in the study of quantum integrable systems, for instance the half-infinite quantum XXZ open spin chain with a non-diagonal integrable boundary. Introduced in \cite{BSh1}, it is called $O_q(\widehat{sl_2})$ which generating elements satisfy the defining relations (\ref{qo1})-(\ref{qo11}). By analogy with Drinfeld's second presentation of $U_q(\widehat{sl_2})$, this new current algebra also admits two different types of presentations: as shown in Theorem 3 of \cite{BSh1} it is isomorphic to the reflection equation algebra with $U_q(\widehat{sl_2})$ $R-$matrix as well as a special case \cite{Ter03} of tridiagonal algebras \cite{Ter93} called the $q-$Onsager algebra\footnote{Note that whereas the reflection equation algebra has been widely studied in the litterature since \cite{Cher84,Skly88}, the algebra (\ref{qDG}) deserves further investigations. More details about mathematical aspects of this algebra and its representation theory can be found in \cite{INT} and references therein.} $\mathbb{T}$ which defining relations are given by (\ref{qDG}). Related with the subject of mathematical physics, looking for an exact solution\footnote{In the spirit of \cite{Onsager,XY,Ahn,Gehlen,Baxter,Gehlen2}, \cite{BK1,BK2} and \cite{Jimbo1}, for instance. In the case of the XXZ open spin chain with {\it generic} integrable boundary conditions, such an alternative to the Bethe ansatz approach is necessary.} to quantum integrable systems with a spectrum generating algebra - or even possibly a hidden symmetry of the Hamiltonian -  of the form (\ref{qDG}) clearly motivates a further study of the algebraic structures (\ref{qo1})-(\ref{qo11}) and (\ref{qDG}).
In view of past experiences, either in mathematics or regarding applications to physics, a better knowledge of the $O_q(\widehat{sl_2})$ current algebra and its {\it explicit} relation with ${\mathbb T}$ is, in particular, highly desirable.\vspace{1mm} 

It is the aim of the present letter to propose an explicit homomorphism relating $O_q(\widehat{sl_2})$ and ${\mathbb T}$.  Besides its own interest in mathematics - one application is considered at the end of this letter -, our result 
finds a straightforward application in the study of quantum integrable systems with a spectrum generating algebra of the form (\ref{qDG}). Indeed, although the existence of a hierarchy of mutually commuting quantities generalizing the Dolan-Grady hierarchy \cite{DG} has been previously conjectured based on the transfer matrix formalism \cite{B1,Bas2,BK,BK1}, a systematic procedure to derive recursively higher elements of the hierarchy as polynomials\footnote{In the case of conformal field theory, mutually commuting elements written as polynomials in the generators of the Virasoro algebra have been derived in \cite{Gervais,Sas}.} in the standard generators ${\textsf A},{\textsf A}^*$ remained an open problem that is solved explicitly in the present letter. 
The interest of the $q-$Dolan-Grady hierarchy presented here relies on the fact that it follows solely from the properties of the algebraic structures involved, contrary to previous approaches which were based on the transfer matrix formalism. As a consequence, whatever the number of dimensions or nature (lattice or continuum) of space-time of the quantum integrable model with spectrum generating algebra ${\mathbb T}$ is, the polynomial structure of the elements of the $q-$Dolan-Grady hierarchy here presented remains unchanged. 
\vspace{1mm} 

The paper is organized as follows. In the next Section, a recursive formula is proposed: given $k$, higher currents' generators $\{{\cW}_{-k},{\cW}_{k+1},{\cG}_{k+1},{\tilde{\cG}}_{k+1}|k\in {\mathbb Z}_+\}$ are written as quadratic combinations of lowest ones. It is shown to be unique, and induces an explicit homomorhism from $O_{q}(\widehat{sl_2})$ to $\mathbb{T}$. Two straightforward applications of our result are then presented in the last Section.  On one hand, using the recursive formula a complete set of mutually commuting quantities that generalize to the $q-$deformed case the Dolan-Grady hierarchy \cite{DG} is exhibited. The first quantities are written explicitely, and agree with the ones previously conjectured from the transfer matrix formalism. At $q=1$, they also coincide with the known results of Dolan-Grady \cite{DG}. On the other hand, the recursive formula allows to derive explicitly {\it higher} Askey-Wilson relations which may serve as defining relations for symmetric special functions generalizing $q-$orthogonal polynomials of the Askey scheme. Note that the extension of our work to other classical Lie algebra - although technically more complicated - is an interesting and open problem. Its starting point - a generalization of (\ref{qDG}) - has been recently introduced in \cite{BB}.
\vspace{0mm}  
\begin{notation}
In this paper, ${\mathbb C}$, ${\mathbb Z}$ denote the field of complex numbers and integers, respectively. We denote ${\mathbb Z}_+$ for nonnegative integers and ${\mathbb C}^*={\mathbb C}/\{0\}$. We introduce the $q-$commutator $\big[X,Y\big]_q=qXY-q^{-1}YX$ where $q$ is the deformation parameter, assumed not to be a root of unity. 
\end{notation}

\section{A recursive formula for the higher generators of $O_{q}(\widehat{sl_2})$}
The defining relations of the new current algebra $O_{q}(\widehat{sl_2})$ follow from the spectral parameter dependent's reflection equation algebra associated with the $U_{q}(\widehat{sl_2})$ $R-$matrix \cite{BSh1}. Similarly to the case of Drinfeld current algebras, the currents can be expanded as formal power series over an infinite set of generators denoted $\{{\cW}_{-k},{\cW}_{k+1},{\cG}_{k+1},{\tilde{\cG}}_{k+1}|k\in {\mathbb Z}_+\}$ which commutation relations were first conjectured in \cite{BK}.
\begin{thm}[see \cite{BSh1}] The current algebra $O_{q}(\widehat{sl_2})$ is isomorphic to the associative algebra with parameter $\rho\in{\mathbb C}^*$, unit $1$ and generators $\{{\cW}_{-k},{\cW}_{k+1},{\cG}_{k+1},{\tilde{\cG}}_{k+1}|k\in {\mathbb Z}_+\}$ satisfying:
\begin{align}
\big[{\cW}_0,{\cW}_{k+1}\big]=\big[{\cW}_{-k},{\cW}_{1}\big]=\frac{1}{(q+q^{-1})}\big({\tilde{\cG}_{k+1} } - {{\cG}_{k+1}}\big)\ ,\label{qo1}\\
\big[{\cW}_0,{\cG}_{k+1}\big]_q=\big[{\tilde{\cG}}_{k+1},{\cW}_{0}\big]_q=\rho{\cW}_{-k-1}-\rho{\cW}_{k+1}\ ,\label{qo2}\\
\big[{\cG}_{k+1},{\cW}_{1}\big]_q=\big[{\cW}_{1},{\tilde{\cG}}_{k+1}\big]_q=\rho{\cW}_{k+2}-\rho{\cW}_{-k}\ ,\label{qo3}\\
\big[{\cW}_{-k},{\cW}_{-l}\big]=0\ ,\quad 
\big[{\cW}_{k+1},{\cW}_{l+1}\big]=0\ ,\label{qo4}\\
%
%
\big[{\cW}_{-k},{\cW}_{l+1}\big]
+\big[{{\cW}}_{k+1},{\cW}_{-l}\big]=0\ ,\label{qo5}\\
\big[{\cW}_{-k},{\cG}_{l+1}\big]
+\big[{{\cG}}_{k+1},{\cW}_{-l}\big]=0\ ,\label{qo6}\\
\big[{\cW}_{-k},{\tilde{\cG}}_{l+1}\big]
+\big[{\tilde{\cG}}_{k+1},{\cW}_{-l}\big]=0\ ,\label{qo7}\\
\big[{\cW}_{k+1},{\cG}_{l+1}\big]
+\big[{{\cG}}_{k+1},{\cW}_{l+1}\big]=0\ ,\label{qo8}\\
\big[{\cW}_{k+1},{\tilde{\cG}}_{l+1}\big]
+\big[{\tilde{\cG}}_{k+1},{\cW}_{l+1}\big]=0\ ,\label{qo9}\\
\big[{\cG}_{k+1},{\cG}_{l+1}\big]=0\ ,\quad   \big[{\tilde{\cG}}_{k+1},\tilde{{\cG}}_{l+1}\big]=0\ ,\label{qo10}\\
\big[{\tilde{\cG}}_{k+1},{\cG}_{l+1}\big]
+\big[{{\cG}}_{k+1},\tilde{{\cG}}_{l+1}\big]=0\ .\label{qo11}
\end{align}
\end{thm}
As will be shown in the Appendix, all higher generators (for $k\geq 1$) can be defined recursively using the set of fundamental commuting relations (\ref{qo1})-(\ref{qo3}). Note that these latter relations together with (\ref{qo4}) at $l=0$ can be actually derived  using the existence of a unique intertwiner of the $q-$Onsager algebra ${\mathbb T}$ (see \cite{BSh1} for details). In the following, we will consider these relations in details. 
\begin{rem} The relations (\ref{qo1})-(\ref{qo3}) may be compared to the ones used to build recursively all generators of the Onsager algebra solely in terms of the fundamental ones $A_0,A_1$. See e.g. \cite{DateRoan}.
\end{rem}
\begin{rem}
Using (\ref{qo1})-(\ref{qo11}), additional commuting relations that will be used in the Apppendix (eqs. (\ref{h1}-\ref{h3})) can be derived. See \cite{BSh1} for details.  
\end{rem}

A natural ordering for the generators of $O_{q}(\widehat{sl_2})$ arises from the study of the commutation relations above. Assume ${\cG}_1,{\tilde{\cG}}_{1}$ are polynomials in the elements ${\cW}_0,{\cW}_1$. Taking $k=0$ in (\ref{qo1}) possible definitions of ${\cG}_1,{\tilde{\cG}}_{1}$ are such that $\mathrm{d}[{\cG}_1]=\mathrm{d}[{\tilde{\cG}}_{1}]\leq 2$, where $\mathrm{d}$ denotes the degree of each monomial in the elements ${\cW}_0,{\cW}_1$. By induction, from (\ref{qo2}), (\ref{qo3}) with (\ref{qo1}) it follows:
\begin{cor} The generators of $O_{q}(\widehat{sl_2})$ are polynomials in ${\cW}_0,{\cW}_1$ of degree:
\beqa
\qquad \qquad \mathrm{d}[{\cW}_{-k}]=\mathrm{d}[{{\cW}}_{k+1}]\leq 2k+1 \qquad \mbox{and} \qquad \mathrm{d}[{\cG}_{k+1}]=\mathrm{d}[{\tilde{\cG}}_{k+1}]\leq 2k+2 \ , \qquad k\in{\mathbb Z}_+. \label{order}
\eeqa
\end{cor}

According to the relations (\ref{qo1})-(\ref{qo3}), given ${\cW}_{-k},{\cW}_{k+1},{{\cG}}_{k+1},{\tilde{\cG}}_{k+1}$ for $k$ fixed, higher elements ${{\cW}}_{-k-1},{{\cW}}_{k+2}$ are uniquely determined. Deriving explicitly ${{\cG}}_{k+2},{\tilde{\cG}}_{k+2}$ in terms of lowest elements is however not so direct in view of (\ref{qo1}). Consider the simplest examples for $k=0,1,2$:
\begin{example} 
The elements ${{\cG}}_{1},{{\cG}}_{2},{{\cG}}_{3}$ can be written as
\beqa
\cG_1&=&[\cW_1,\cW_0]_q+a_1\ ,\nonumber\\
\cG_2&=&\frac{1}{(q^2+q^{-2})}\left(q\, [\cW_2,\cW_0]_{q^2}+q^{-1}\,  [\cW_1,\cW_{-1}]_{q^2}-(q-q^{-1})\,\Big(q^{-2} \,(\cW_1)^2+q^2 \,(\cW_0)^2+\frac{\cG_1\,\tilde{\cG}_1}{\rho}\Big)\right)+a_2\ ,\nonumber\\
\cG_3&=&\frac{1}{(q^2+q^{-2}-1)}\left((q-q^{-1})\,\big(q^2\,\cW_3\,\cW_0+q^{-2}\,\,\cW_{-2}\,\cW_1\big)+[\cW_2,\cW_{-1}]_q\right.\nonumber\\
&&\qquad \qquad \qquad \qquad \left. -(q-q^{-1})\,\Big(q^{-2}\,\cW_1 \,\cW_2+q^2 \,\cW_0\,\cW_{-1}+\frac{\cG_2\,\tilde{\cG}_1}{\rho}\Big)\right) +a_3\ ,\nonumber
\eeqa
where $a_1,a_2,a_3\in {\mathbb C}$ are arbitrary. Expressions of ${\tilde{\cG}}_{1},{\tilde{\cG}}_{2},{\tilde{\cG}}_{3}$ are obtained exchanging ${\cW}_{-k}\leftrightarrow{\cW}_{k+1}$ and ${\cG}_{k+1}\leftrightarrow{\tilde{\cG}}_{k+1}$ in above expressions.
\end{example}
\begin{proof}
For $k=0,1,2$, we have to show that writting ${\cG}_{k+1}$ (and similarly ${\tilde\cG}_{k+1}$) as the most general quadratic combination of lowest elements such that $\mathrm{d}[{\cG}_{k+1}]\leq 2k+2$, up to the relations (\ref{qo1})-(\ref{qo11}) the coefficients in the combination are uniquely determined.  By inspection of the structure and degree's dependency in the relations (\ref{qo1})-(\ref{qo3}) and (\ref{qo4}) for $l=0$, only quadratic combinations of lowest elements of even degree can be actually considered. The derivation of the case $k=0$ is obvious. For the case $k=1$, the most general possible combination includes lowest terms of the form $\cW_0\cW_2$, $\cW_2\cW_0$, $\cW_1\cW_{-1}$, $\cW_{-1}\cW_{1}$, $\cW_0^2$, $\cW_1^2$,  $\cG_1\,\tilde{\cG}_1$ and $\tilde{\cG}_1\cG_1$. Noticing that\footnote{Generalizations of this equation are reported in the Appendix, see eqs. (\ref{h1}-\ref{h3}).} 
\beqa
-{\cW}_{0}^2 + {\cW}_{1}^2 - {\cW}_{-1}{\cW}_{1} + {\cW}_{0}{\cW}_{2} - \frac{1}{\rho(q^2-q^{-2})}\big[{\cG}_{1},{\tilde{\cG}}_{1}\big]=0\ ,
\eeqa
the equations (\ref{qo1})-(\ref{qo3}) are reduced to a set of irreducible equations, leading to linear constraints determining uniquely the coefficients.
The derivation of the case $k=2$ is similar. In each case, a straightforward calculation shows that other relations (\ref{qo4}) for $l\neq 0$ and (\ref{qo5})-(\ref{qo11}) are automatically satisfied: no additional constraints on the coefficients appear. Although not reported here, up to $k=5$ a similar structure for the elements ${{\cG}}_{k+1},{\tilde{\cG}}_{k+1}$ is observed.
\end{proof}

Inspired by the examples above, higher generators ${\cG}_{k+1}$ (and similarly for ${\tilde\cG}_{k+1}$) can be derived explicitly in terms of lower ones along the same lines. The proof of the following result is reported in the Appendix.

\begin{prop}{\label{prop1}} For $i,j\in {\mathbb Z}_+$, define $\bar{k}=1$ (resp. $0$) for $k$ even (resp. odd)  and $\alpha=\null[\frac{k}{2}]=\null \frac{k}{2}$ (resp. $\frac{k-1}{2}$)  for   $k$ even (resp. odd).  Given $k$, the highest elements can be written 
\beqa
\cG_{k+1}&=&\sum_{l=0}^{\alpha}\sum_{i+j=2l+1-\bar{k}} A_{ij} + \sum_{l=0}^{\alpha-\bar{k}}\sum_{i+j=2l+\bar{k}, i\leq j} F_{ij} \ + \ a_k \ ,\label{ansatz}\\
A_{ij}&=&a^{(k)}_{ij} \cW_{-i}\cW_{j+1}+b^{(k)}_{ij} \cW_{i+1} \cW_{-j} \ , \nonumber \\
F_{ij}&=&e^{(k)}_{ij}\big(q^2 \cW_{-i}\cW_{-j}+q^{-2}\cW_{i+1}\cW_{j+1}+\frac{1}{\rho} \cG_{j+1}\tilde{\cG}_{i+1}\big)\nonumber\ .
\eeqa
The coefficients $a^{(k)}_{ij} ,b^{(k)}_{ij}$ and $e^{(k)}_{ij}$ $\in {\mathbb{C}}$ are uniquely determined up to (\ref{qo1}-\ref{qo11}), and $a_k\in {\mathbb{C}}$ is an arbitrary parameter. The explicit expression of ${\tilde{\cG}}_{k+1}$ is obtained exchanging ${\cW}_{-k}\leftrightarrow{\cW}_{k+1}$ and ${\cG}_{k+1}\leftrightarrow{\tilde{\cG}}_{k+1}$ in the expression above.
\end{prop}

Thanks to the recursive formula (\ref{ansatz}), we are now in position to exhibit an explicit homomorphism relating the current algebra $O_{q}(\widehat{sl_2})$ and the $q-$Onsager algebra $\mathbb{T}$. Recall that the $q-$Onsager algebra $\mathbb{T}$ is a special case of tridiagonal algebras that have been introduced and studied in
\cite{Ter93,Ter01,Ter03}, where they first appeared in the context
of $P-$ and $Q-$polynomial association schemes. A tridiagonal algebra is an associative algebra with unit which consists of two generators ${\textsf A}$ and ${\textsf A}^*$ called the standard generators. In general, the defining relations depend on five scalars $\rho,\rho^*,\gamma,\gamma^*$ and $\beta$. Below, we focus on the {\it reduced} parameter sequence $\gamma=0,\gamma^*=0$, $\beta=q^2+q^{-2}$ and $\rho=\rho^*$  which exhibits all interesting properties that can be extended to more general parameter sequences.
We call the corresponding algebra the $q-$Onsager algebra denoted ${\mathbb T}$, in view of its close relationship with the Onsager algebra \cite{Onsager} and the Dolan-Grady relations \cite{DG}. Note that the isomorphism between the Onsager and Dolan-Grady  algebraic structures has been studied in \cite{Pe,APMPT,Da} and shown explicitly in \cite{DateRoan}.
\begin{defn}[see also \cite{Ter03}]
The $q-$Onsager algebra $\mathbb{T}$ is the associative algebra with unit and standard generators $\textsf{A},\textsf{A}^*$ subject to the following relations
\beqa
[\textsf{A},[\textsf{A},[\textsf{A},\textsf{A}^*]_q]_{q^{-1}}]=\rho[\textsf{A},\textsf{A}^*]\
,\qquad
[\textsf{A}^*,[\textsf{A}^*,[\textsf{A}^*,\textsf{A}]_q]_{q^{-1}}]=\rho[\textsf{A}^*,\textsf{A}]\
\label{qDG} . \eeqa
\end{defn}
\begin{rem} For $\rho=0$ the relations (\ref{qDG})
reduce to the $q-$Serre relations of $U_{q}(\widehat{sl_2})$. For $q=1$, $\rho=16$ they coincide with the Dolan-Grady relations \cite{DG}.
\end{rem}
\begin{prop}{\label{prop2}} The recursive relations (\ref{qo2}), (\ref{qo3}) together with Proposition 2.1 induce an explicit homomorphism from $O_{q}(\widehat{sl_2})$ to $\mathbb{T}$. For instance, 
\beqa
{\cW}_{0}\rightarrow \textsf{A}\ ,\qquad {\cW}_{1}\rightarrow \textsf{A}^* \label{homo}\ .
\eeqa
\end{prop}
\begin{proof} According to the isomorphism [see \cite{BSh1}, Theorem 3] between $O_{q}(\widehat{sl_2})$ and $\mathbb{T}$, any element of $O_{q}(\widehat{sl_2})$ can be realized as a non-linear combination of the standard generators $\textsf{A},\textsf{A}^*$.  According to (\ref{qo4}) for $l=0,k=1$, we choose for instance (\ref{homo}) - another obvious possiblity being $\textsf{A}\leftrightarrow \textsf{A}^*$. Replacing (\ref{homo}) in (\ref{qo2}), (\ref{qo3}) and applying (\ref{ansatz}) for $k=0,1,2,...$ an explicit realization of higher elements in terms of $\textsf{A},\textsf{A}^*$ follow.    
\end{proof}
\begin{example} 
The next elements read:
\beqa
&&{\cG}_{1} \rightarrow \big[{{\textsf A}^*},{\textsf A}\big]_q + a_1 \ ,\label{defel}\\
&&{\cW}_{-1}\rightarrow\frac{1}{\rho}\left( (q^2+q^{-2})\textsf{A}\textsf{A}^*\textsf{A} -\textsf{A}^2\textsf{A}^* - \textsf{A}^* \textsf{A}^2\right) + \textsf{A}^* + \frac{a_1(q-q^{-1})}{\rho}\textsf{A}\ ,\nonumber\\
&&{\cG}_{2}\rightarrow \frac{1}{\rho(q^2+q^{-2})} \left( (q^{-3}+q^{-1}) {\textsf A}^2{{\textsf A}^*}^2 - (q^{3}+q){{\textsf A}^*}^2{\textsf A}^2 + (q^{-3}-q^{3})({\textsf A}{{\textsf A}^*}^2{\textsf A} + {{\textsf A}^*}{\textsf A}^2{{\textsf A}^*}) \right. \nonumber\\
&&\qquad \qquad \left. - (q^{-5}+q^{-3} +2q^{-1}) {\textsf A}{{\textsf A}^*}{\textsf A}{{\textsf A}^*} + (q^{5}+q^{3} +2q){{\textsf A}^*}{\textsf A}{{\textsf A}^*}{\textsf A} +  \rho(q-q^{-1})({\textsf A}^2 + {{\textsf A}^*}^2
)\right. \nonumber\\
&&\qquad \qquad \left.+a_1(q^2+q^{-2})(q-q^{-1})\big[{{\textsf A}^*},{\textsf A}\big]_q \right) + a_2 \ ,\nonumber 
\eeqa
where $a_1,a_2\in {\mathbb C}$ are arbitrary. Expressions of ${\cW}_{k+1}$ and ${\tilde{\cG}}_{k+1}$ are obtained from ${\cW}_{-k}$ and ${{\cG}}_{k+1}$ exchanging ${\textsf A}\leftrightarrow{\textsf A}^*$ .
\end{example}
Note that the expressions above agree with the ones proposed in \cite{BK}. However, contrary to \cite{BK} the expressions here are not conjectured using the properties of certain finite dimensional tensor product representations of (\ref{qDG}). 

\section{Applications}
In \cite{DG}, Dolan and Grady studied the existence of an infinite family of mutually commuting quantities in quantum integrable systems with an underlying spectrum generating algebra of the form (\ref{qDG}) at $q=1$. A one-parameter family of mutually commuting elements was constructed recursively, which first element coincides with the Hamiltonian of the integrable system under consideration - for instance the planar Ising or superintegrable chiral Potts model.
For $q\neq 1$, the existence of a $q-$deformed analog of the Dolan-Grady hierarchy was actually conjectured in \cite{Bas2}, but writting explicitly higher mutually commuting elements in terms of ${\textsf A},{\textsf A}^*$ remained technically problematic. To circumvent this difficulty, the transfer matrix formalism was used: starting from the solutions of the reflection equation (\cite{B1,BK}), combinations of the form\footnote{The transfer matrix of the XXZ open spin chain with non-diagonal boundary conditions is a linear combination of such quantities \cite{BK1}.} 
\beqa
{\cal I}_{2k+1}=\kappa {\cal W}_{-k} + \kappa^* {\cal W}_{k+1} + \kappa_+ {\cal G}_{k+1} 
+ \kappa_- {\tilde{\cal G}}_{k+1}\ \label{Icons} 
\eeqa
were found to be mutually commuting for arbitrary parameters $\kappa,\kappa^*,\kappa_\pm\in {\mathbb C}$. Based on the analysis of certain finite dimensional tensor product representations of (\ref{qo1})-(\ref{qo11}), the hierarchy (\ref{Icons}) was then conjectured to be in one-to-one correspondence with a $q-$deformed analog of the Dolan-Grady hierarchy. For these representations, explicit calculations up to $k=3$ supported the conjecture \cite{BK}. But until the isomorphism between $O_q(\widehat{sl_2})$ and ${\mathbb T}$ was shown (see \cite{BSh1}), the exact relation between both hierarchies at the level of the algebra could not be investigated.\vspace{1mm}

Thanks to the results of \cite{BSh1} and the homomorphism here proposed, we are now in position to construct the $q-$Dolan Grady hierarchy - with no reference to a representation space on which ${\textsf A},{\textsf A}^*$ act, contrary to previous works \cite{BK}. Indeed, it is easy to check solely using (\ref{qo1})-(\ref{qo11}) that the quantities (\ref{Icons}) are mutually commuting.  Consider for instance the simplest $q-$deformed analog of the Dolan-Grady hierarchy, i.e. $\kappa_\pm=0$ and choose $a_1=a_2=0$ in (\ref{defel}). A simple calculation leads to the following hierarchy of mutually commuting quantities:
\beqa
{\cal I}_{2k+1}=\kappa f_k({\textsf A},{\textsf A}^*) + \kappa^* f_k({\textsf A}^*,{\textsf A}) \label{charge}
\eeqa
where the polynomials $f_k({\textsf A},{\textsf A}^*)$ are computed recursively using (\ref{qo2})-(\ref{qo3}) and Proposition 2.2. For instance,
\beqa
f_0({\textsf A},{\textsf A}^*)&=& {\textsf A}\ ; \qquad f_1({\textsf A},{\textsf A}^*)=\frac{1}{\rho}\left( (q^2+q^{-2})\textsf{A}\textsf{A}^*\textsf{A} -\textsf{A}^2\textsf{A}^* - \textsf{A}^* \textsf{A}^2\right) + \textsf{A}^* \ ,\nonumber\\
f_2({\textsf A},{\textsf A}^*)&=&\frac{1 }{\rho^2(q^2+q^{-2})}\left(  (q^{-2}+1)\textsf{A}^3{\textsf{A}^*}^2 + (q^{2}+1){\textsf{A}^*}^2{\textsf{A}}^3 -(q^{4}+q^{-4})(\textsf{A}{\textsf{A}^*}^2 \textsf{A}^2 + \textsf{A}^2{\textsf{A}^*}^2 \textsf{A})  \right.\nonumber\\
&&\qquad \qquad \qquad \left. + (q^{-2}-q^{4})\textsf{A}{\textsf{A}}^*\textsf{A}^2\textsf{A}^*+(q^{2}-q^{-4})\textsf{A}^*{\textsf{A}}^2\textsf{A}^*\textsf{A} \right.\nonumber\\
&&\qquad \qquad \qquad \left.- (q^{-4}+q^{-2}+2)\textsf{A}^2{\textsf{A}}^*\textsf{A}\textsf{A}^* - (q^{4}+q^{2}+2)\textsf{A}^*\textsf{A}\textsf{A}^*\textsf{A}^2 \right.\nonumber\\
&&\qquad \qquad \qquad \left. - (q^{6} + q^{4}+2q^{2}+2q^{-2}+q^{-4}+q^{-6})\textsf{A}\textsf{A}^*\textsf{A}\textsf{A}^*\textsf{A}\right. \nonumber\\
&&\qquad \qquad \qquad \left. + (q^2-1)(\textsf{A}^3+\textsf{A}{\textsf{A}^*}^2) + (q^{-2}-1)(\textsf{A}^3+\textsf{A}^*{\textsf{A}}^2)\right) + f_1({\textsf A}^*,{\textsf A})\nonumber \ .
\eeqa
Setting $\rho=16$ and $q=1$ in (\ref{qDG}) one recovers the Dolan-Grady relations, in which case the first few elements simply reduce to the ones proposed in \cite{DG}, as expected.\vspace{1mm} 

As suggested in \cite{Ter03} the representation theory of the $q-$Onsager algebra (\ref{qDG}) potentially provides a classification scheme for special symmetric fonctions, some of them being already well-known. In this context, the recursive formula (\ref{ansatz}) finds another straightforward application. Among the simplest examples, consider the elements of (\ref{qDG}) realized as $\textsf{A}\rightarrow A,\textsf{A}^*\rightarrow A^*$ where  $A,A^*$ satisfy the Askey-Wilson algebra \cite{Zhed} with defining relations:  
\beqa A^*A^2 + {A^2}A^* - (q^2+q^{-2})AA^*A  -
\rho A^* - \omega A&=& 0 \  ,\label{nonstand}\\
A{A^*}^2 + {A^*}^2A - (q^2+q^{-2})A^*AA^*  -
{\rho} A - \omega A^* &=& 0\ .\nonumber 
\ 
\eeqa
Remarkably, the
$q-$orthogonal Askey-Wilson polynomials with variable $x\equiv z+z^{-1}$ defined by\footnote{For arbitrary parameters $a,b,c,d$ such that none of the combinations $ab,ac,ad,bc,bd,cf,abcd$ are integer powers of $q$. Here, $_4\Phi_3$ denotes the basic $q-$hypergeometric function.}
\beqa
p_n(x;a,b,c,d)= _4\!\!\Phi_3 \left[ \begin{array}{c} q^{-n},\ abcdq^{n-1},\ az,\ az^{-1} \\
ab, \ ac, \ ad   
\end{array} ;q \ | \ q\right]
\eeqa
provide an infinite dimensional representation\footnote{Note that relations between the Askey-Wilson algebra and Cherednik's double affine Hecke algebra have been considered for instance in \cite{Koo}.} of (\ref{nonstand}) in which case the element $A$ (resp. $A^*$) acts as a second-order $q-$difference operator (resp. $z+z^{-1}$) in the variable $z$ (see e.g. \cite{Zhed,Gra,Gru,Nou,Ter03} ). Note that finite dimensional representations can also be obtained by restricting the variable $z$ to a discret support, in which case 
$A,A^*$ are identified to Leonard pairs \cite{Terleo}.\vspace{1mm}

According to the homomorphism (\ref{homo},\ref{defel}), it is clear that the Askey-Wilson relations (\ref{nonstand}) can be alternatively written as:
\beqa
{\cal W}_{-1} + \alpha {\cal W}_{0}=0\ \qquad \mbox{et} \qquad {\cal W}_{2} + \alpha {\cal W}_{1}=0\ \label{AW}
\eeqa
where $\alpha=(\omega-a_1(q-q^{-1}))/\rho$. It is then interesting to recall that
based on the analysis of finite dimensional tensor product representations of $O_q(\widehat{sl_2})$, generalizations of these relations were derived in \cite{BK} (see eqs. (54-57)). In particular, for the simplest generalization of (\ref{AW}) the elements ${\cal W}_{-2},{\cal W}_{-1},{\cal W}_{2},{\cal W}_{3}$ of $O_q(\widehat{sl_2})$ were realized as  polynomials in a tridiagonal pair\footnote{For a definition of tridiagonal pairs and the Askey-Wilson relations, see for instance \cite{IT}.}  associated with ${\textsf A},{\textsf A}^*$. Using the relation between the reflection equation algebra and the algebra $O_q(\widehat{sl_2})$ (see for instance \cite{Bas2,BK}), generalizations of (\ref{AW}) are clearly expected for arbitrary finite dimensional representations. By analogy with the Askey-Wilson algebra, it is thus natural to consider the {\it $N-$th order higher Askey-Wilson algebra} which defining relations can be written in the simple form:
\beqa
\sum^{N}_{k=0}\alpha_{k}^{(N)}{\cal W}_{-k}=0\ \qquad \mbox{and} \qquad \sum^{N}_{k=0}\alpha_{k}^{(N)}{\cal W}_{k+1}=0\ ,\label{genAW}
\eeqa
where the coefficients $\alpha_{k}^{(N)}\in {\mathbb C}$ depend on the representation (finite or infinite dimensional) on which the elements ${\textsf A},{\textsf A}^*$ act. Here, each element ${\cal W}_{-k},{\cal W}_{k+1}$ is considered as a polynomial in ${\textsf A},{\textsf A}^*$ thanks to the recursive formula (\ref{ansatz}) and using the homomorphism (\ref{homo}). 
For finite dimensional representations, explicit examples can be found in \cite{BK,BK1}. Starting from (\ref{genAW}), infinite dimensional representations of the $q-$Onsager algebra and related special functions can now be studied systematically using explicit realizations  of ${\textsf A},{\textsf A}^*$ - generalizing the Askey-Wilson one - in terms of $q-$difference operators acting on the space of symmetric functions. This subject which finds interesting application in the context of quantum integrable systems and Bethe equations will be discussed elsewhere.

\vspace{1cm}

\centerline{\underline{APPENDIX:} Proof of Proposition \ref{prop1}}

\vspace{0.5cm}

Having in mind the ordering (\ref{order}) and the fundamental relations (\ref{qo1}-\ref{qo11}), according to Example 1 and the results up to $k=5$ (not reported here), it is natural to propose for ${\cG}_{k+1}$ (similarly for ${\tilde\cG}_{k+1}$) the most general quadratic combination of lowest elements of even degree such that $\mathrm{d}[{\cG}_{k+1}]\leq 2k+2$. Namely, assume the relations (\ref{qo1}-\ref{qo11}) hold for $i,j\in {\mathbb Z_+}$. A linear combinations of terms $\cW_{-i}\cW_{j+1}$, $\cW_{i+1}\cW_{-j}$, $\cW_{-i}\cW_{-j}$, $\cW_{i+1}\cW_{j+1}$ with $i+j\leq k$ and $\cG_{i+1}\cG_{j+1}$, ${\tilde\cG}_{i+1}{\tilde\cG}_{j+1}$, ${\cG}_{i+1}{\tilde\cG}_{j+1}$, ${\tilde\cG}_{i+1}{\cG}_{j+1}$ with $i+j\leq k-1$ can then be considered in full generality, which explains the structure of the proposal (\ref{ansatz}). Now, the proof that (\ref{ansatz}) satisfy the relations (\ref{qo1}-\ref{qo11}) goes in two steps.

First, we show that (\ref{ansatz}) - and corresponding expression for $\tilde{\cG}_{k+1}$ - satisfy (\ref{qo1}-\ref{qo2}) - (\ref{qo1}-\ref{qo3}), respectively - and the coefficients $\{a^{(k)}_{ij},b^{(k)}_{ij},e^{(k)}_{ij}\}$ are uniquely determined. To this end, we will need the following relations which can be derived from (\ref{qo1})-(\ref{qo11}) (see \cite{BSh1} for details):
\ben
&&\big[{\cW}_{-i-1},{\cW}_{j+1}\big] - \big[{\cW}_{-i},{\cW}_{j+2}\big] = \frac{q-q^{-1}}{\rho(q+q^{-1})}\big({\cG}_{i+1}\tilde{{\cG}}_{j+1}-{\cG}_{j+1}\tilde{{\cG}}_{i+1}\big)\ ,\label{h1}\\
&&-{\cW}_{-i}{\cW}_{0} + {\cW}_{i+1}{\cW}_{1} - {\cW}_{-i-1}{\cW}_{1} + {\cW}_{0}{\cW}_{i+2} - \frac{1}{\rho(q^2-q^{-2})}\big[{\cG}_{i+1},{\tilde{\cG}}_{1}\big]=0\ ,\label{h2}\\
&&{\cW}_{-i-1}{\cW}_{-j} - {\cW}_{i+2}{\cW}_{j+1} - {\cW}_{-i}{\cW}_{-j-1} + {\cW}_{i+1}{\cW}_{j+2}\label{h3}\\
&&+{\cW}_{-i}{\cW}_{j+1} - {\cW}_{-j}{\cW}_{i+1} - {\cW}_{-i-1}{\cW}_{j+2} + {\cW}_{-j-1}{\cW}_{i+2}\nonumber\\
&& + \frac{1}{\rho(q^2-q^{-2})}
\big(\big[{\cG}_{i+2},{\tilde{\cG}}_{j+1}\big]-\big[{\cG}_{i+1},{\tilde{\cG}}_{j+2}\big]\big)=0\ .\nonumber
\een
In order to rewrite recursion relations, introduce the new operators
\ben
\cB_{ij}&=&\cW_{i+1}\cW_{-j}-\cW_{j+1}\cW_{-i}=\cW_{-j}\cW_{i+1}-\cW_{-i}\cW_{j+1}\\
\cC_{ij}&=&[\cW_{-i},\cW_{j+1}]=[\cW_{-j},\cW_{i+1}] \\
\cD_{ij}&=&\cG_{i+1}\tilde{\cG}_{j+1}-\cG_{j+1}\tilde{\cG}_{i+1}=\tilde{\cG}_{j+1}\cG_{i+1}-\tilde{\cG}_{i+1}\cG_{j+1}\\
\cF_{ij}&=&\cW_{i+1}\cW_{j+1}- \cW_{-i}\cW_{-j}+w(\tilde \cG_{j+1}\cG_{i+1}- \cG_{j+1}\tilde{\cG}_{i+1}) \ .
\een
Note that $\cB_{ij},\cD_{ij}$ are antisymmetric and $\cC_{ij}$ is symmetric. In terms of these, the recursion relations (\ref{h1})-(\ref{h3}) read
\ben
\cC_{i+1,j}&=&\cC_{i,j+1}+v\cD_{ij}\ ,\nonumber\\
\cF_{0,j}&=&\cB_{0j+1}+w\cD_{0j} \quad \mb{for} 0\leq j \ ,  \nonumber\\
\cB_{ij}&=&\cB_{i+1,j+1}+\cF_{i,j+1}-\cF_{i+1,j}+w(\cD_{i+1,j}+\cD_{j+1,i})\ ,\label{com}\\
 v&=&(q-q^{-1})^2 w=\frac{q-q^{-1}}{\rho(q+q^{-1})}\ .\nonumber
\een
First, let us consider equation (\ref{qo1}). Using the above combinations $\{\cB,\cC,\cF\}$ and (\ref{ansatz}),  eq. (\ref{qo1}) becomes:
\ben
(q+q^{-1})\cC_{0k}&=&\sum_{l=0}^{\alpha}\Big(\sum_{i=0}^{l-\bar{k}}(a^-_{i,2l+1-\bar{k}-i}\cB_{i,2l+1-\bar{k}-i}- a^+_{i,2l+1-\bar{k}-i}\cC_{i,2l+1-\bar{k}-i})-\delta_{\bar{k},1}\frac{a^+_{ll}}{2}\cC_{ll}\Big)\nonumber\\
&&+ (q^2-q^{-2})\sum_{l=0}^{\alpha-\bar{k}}\Big(\sum_{i=0}^{l} e^{(k)}_{i,2l+\bar{k}-i}\cF_{i,2l+\bar{k}-i}\Big)\nonumber
\een
with
\ben
a^\pm_{ij}&=&(a^{(k)}_{ij}-b^{(k)}_{ij})\pm(a^{(k)}_{ji}-b^{(k)}_{ji})\ .\nonumber
\een
Applying recursively the relations (\ref{com}), the parameters $\{a^{(k)}_{ij},b^{(k)}_{ij},e^{(k)}_{ij}\}$ are found to be restricted by  
the following constraints, with $l \in \{1,\dots,\alpha\}$:
\ben
\label{CS1}
&&\fs_l+\delta_{l\alpha}(q+q^{-1})=0, \quad \fp_{0l}-w(q^2-q^{-2})e^{(k)}_{0,2l-\bar{k}}=0\,, \quad  \bar{a}_{ik-i}=0 \mb{for} i \in \{0,\dots,\alpha-\bar{k}\}\ ,\nonumber\\
&&  \bar{a}_{l-1,l+\bar{k}}-(q^2-q^{-2})e^{(k)}_{l,l+\bar{k}}=0, \quad \delta_{\bar{k},1}(w\bar{a}_{l-1,l+1}-\fp_{l,l+1})=0\ , \\
&& w(\bar{a}_{i-1,2l+\bar{k}-i}-\bar{a}_{i,2l-1+\bar{k}-i})-\fp_{il+\bar{k}}=0  \mb{for}  0<i<l-\bar{k}\ ,\nonu
 &&\bar{a}_{i,2l-1-i}-\bar{a}_{i-1,2l+\bar{k}-i} +(q^2-q^{-2})e^{(k)}_{i,2l+\bar{k}-i}=0 \mb{for}  0<i<l-\bar{k}\ \nonumber
\een
where
\ben
&&\fs_l=\sum_{i=0}^{l-\bar{k}}\fa^+_{i,2l+1-\bar{k}-i}+\delta_{\bar{k},1}\frac{\fa^+_{l,l}}{2}\ , \quad \fp_{i,l}=v\big(\sum_{j=i+1}^{l-\bar{k}}\fa^+_{j,2l+1-\bar{k}-i}+\delta_{\bar{k}1}\frac{\fa^+_{l,l}}{2}\big)\ , \label{defa}\\
&&\bar{a}_{ij}=a^-_{ij}+\bar{a}_{i-1,j-1} \mb{for} 0<i<j \,,\nonu
&&\bar{a}_{0i}=a^-_{0i}+(q^2-q^{-2})e^{(k)}_{0,i-1}+\bar{a}_{0,i-2}\quad \mbox{for} \quad i>1+\bar{k} \ ,\nonu
&&
\bar{a}_{0,1+\bar{k}}= a^-_{0,1+\bar{k}}+(q^2-q^{-2})e^{(k)}_{0,\bar{k}}\ .\nonumber
\een

Consider now the equation (\ref{qo2}). Rewritting the $q-$commutator
\ben
[A,B]_q-[C,A]_q&=&(q-q^{-1})(B-C)A+[A,qB+q^{-1}C]\ ,
\een
one has
\ben
(q-q^{-1})(\cG_{k+1}-\tilde{\cG}_{k+1})\cW_0+[\cW_0,q\cG_{k+1}+q^{-1}\tilde{\cG}_{k+1}]=0\ .\label{eqgg}
\een
Introduce the operators
\ben
\A_{ij}=q\cA_{ij}+q^{-1}\tilde{\cA_{ij}},\quad
\FF_{ij}=q\cF_{ij}+q^{-1}\tilde{\cF_{ij}}\mb{and}
\GG_{i+1}=\tilde{\cG}_{i+1}-\cG_{i+1}\nonumber \ .
\een
Using (\ref{ansatz}) (and similarly for  $\tilde{\cG}_{k+1}$),  eq. (\ref{eqgg}) becomes:
\ben
(q-q^{-1})\GG_{k+1}\cW_0=\sum_{l=0}^{\alpha} \sum_{i+j=2l+1-\bar{k}}[\cW_0,\A_{ij}]+\sum_{l=0}^{\alpha-\bar{k}}\sum_{i+j=2l+\bar{k}, i\leq j}[\cW_0,\FF_{ij}\null]\ . \label{S2}
\een
According to the commutation relations (\ref{qo1})-(\ref{qo4}),  one has:
\ben
\null[\cW_0,\FF_{ij}]&=&e^{(k)}_{ij}(\cW_{-j-1}\GG_{i+1}+\GG_{j+1}\cW_{-i-1})\ ,\nonumber\\
\null [\cW_0,\A_{ij}]&=&\alpha_{ij}\cW_{-i}\GG_{j+1}+\bar{\alpha}_{ij}\GG_{i+1}\cW_{-j}\ 
\label{comS2}
\een
with
\ben
\alpha_{ij}=\frac{qa^{(k)}_{ij}+q^{-1} b^{(k)}_{ij}}{q+q^{-1}} \mb{and}
\bar{\alpha}_{ij}=\frac{q b^{(k)}_{ij}+q^{-1} a^{(k)}_{ij}}{q+q^{-1}}\ .\label{exp}
\een
Regrouping all terms of the form $\cW_{-j}\GG_{i+1}$ and $\GG_{j+1}\cW_{-i}$, eq. (\ref{qo2}) imposes additional constraints on the parameters $\{a^{(k)}_{ij},b^{(k)}_{ij},e^{(k)}_{ij}\}$. With $l \in \{1,\dots,\alpha\}$, they read:
\ben
\label{CS2}
&&\alpha_{00}=\bar{\alpha}_{00}=0\,, \, \bar{\alpha}_{2l+1-\bar{k},0}=\delta_{l,\alpha}(q-q^{-1})\ , \nonumber\\
&&\alpha_{2l+1-\bar{k}-i,i}+e^{(k)}_{i,2l-\bar{k}-i}=0 \, \mb{for} \,  i \in \{0,\dots,l-\bar{k}\} \ ,\nonu
&&\alpha_{2l+1-\bar{k}-i,i}=0 \, \mb{for} \,  i \in \{l+1-\bar{k},\dots,2l+1-\bar{k}\}\ ,\\
&&\bar{\alpha}_{2l+1-\bar{k}-i,i}+e^{(k)}_{i-1,2l+1-\bar{k}-i}=0 \, \mb{for} \,  i \in \{1,\dots,l+1-\bar{k}\} \mb{and} \nonu
&&\bar{\alpha}_{2l+1-\bar{k}-i,i}=0 \, \mb{for}  i \in \{l+2-\bar{k},\dots,2l+1-\bar{k}\}\ .\nonumber
\een

Now, combining the first (\ref{CS1}) and second (\ref{CS2}) set of constraints, the parameters $\{a^{(k)}_{ij},b^{(k)}_{ij}\}$ can be written solely in terms of $\{e^{(k)}_{ij}\}$ :
\ben
a^{(k)}_{2l+1-\bar{k},0}&=&-q^{-1}\delta_{l,\alpha}-\frac{q}{q-q^{-1}}e^{(k)}_{0,2l-\bar{k}}\ ,\quad
b^{(k)}_{2l+1-\bar{k},0}=q\delta_{l,\alpha}+\frac{q^{-1}}{q-q^{-1}}e^{(k)}_{0,2l-\bar{k}}\ ,\nonu
a^{(k)}_{2l+1-\bar{k}-i,i}&=&-\frac{1}{q-q^{-1}}\left(q e^{(k)}_{i,2l-\bar{k}-i}-q^{-1}e^{(k)}_{i-1,2l+1-\bar{k}-i}\right)  \mb{for}  i \in \{1,\dots,l-\bar{k}\}\ ,\nonu
b^{(k)}_{2l+1-\bar{k}-i,i}&=&-\frac{1}{q-q^{-1}}\left(qe^{(k)}_{i-1,2l+1-\bar{k}-i}-q^{-1}e^{(k)}_{i,2l-\bar{k}-i}\right) \mb{for}  i \in \{1,\dots,l-\bar{k}\}\ ,\\
a^{(k)}_{2l+1-\bar{k}-i,i}&=&b^{(k)}_{2l+1-\bar{k}-i,i}=0 \mb{for}  i \in \{l+2-\bar{k},\dots,2l+1-\bar{k}\}\ ,\nonu
a^{(k)}_{l,l+1-\bar{k}}&=&\frac{q^{-1}}{q-q^{-1}}e^{(k)}_{l-\bar{k},l},\quad b^{(k)}_{l,l+1-\bar{k}}=-\frac{q}{q-q^{-1}}e^{(k)}_{l-\bar{k},l}\ .\nonumber
\een
Plugging the expressions above in the definitions (\ref{exp}), we write:
\ben
a^\pm_{0,2l+1-\bar{k}}&=&\mp\left((q+q^{-1})\delta_{l,\alpha}+\frac{q+q^{-1}}{q-q^{-1}}e^{(k)}_{0,2l-\bar{k}}\right)\nonu
a^{\pm}_{2l+1-\bar{k}-i,i}&=&\mp\frac{q+q^{-1}}{q-q^{-1}}\left(e^{(k)}_{i,2l-\bar{k}-i}-e^{(k)}_{i-1,2l+1-\bar{k}-i}\right) \mb{for}  i \in \{1,\dots,l-1\},\nonu
a^{+}_{l,l+1-\bar{k}}&=&(1+\bar{k})\frac{q+q^{-1}}{q-q^{-1}}e^{(k)}_{l-\bar{k},l+1-\bar{k}}\nonu
a^{-}_{l,l+1-\bar{k}}&=&\delta_{\bar{k}1}\frac{q+q^{-1}}{q-q^{-1}}\left(2e^{(k)}_{l,l}-e^{(k)}_{l-1,l+1}\right)\ .\nonumber
\een 
According to (\ref{defa}), a straightforward calculation shows that the summation of coefficients $a^\pm_{i,j}$ drastically simplifies (\ref{CS1}). It reduces to the system of equations for $l \in \{1,\dots,\alpha\}$: 
\ben
&& \bar{a}_{ik-i}=0 \mb{for} i \in \{0,\dots,\alpha-\bar{k}\}\ ,\nonumber\\
&&  \bar{a}_{l-1,l+\bar{k}}-(q^2-q^{-2})e^{(k)}_{l,l+\bar{k}}=0\ , \label{sol} \\
\mb{and} &&  \bar{a}_{i,2l-1-i}-\bar{a}_{i-1,2l+\bar{k}-i} +(q^2-q^{-2})e^{(k)}_{i,2l+\bar{k}-i}=0 \mb{for}  0<i<l-\bar{k}\ . \nonumber
\een
Having in mind (\ref{defa}), the number of independent equations $\frac{(\alpha-\bar{k}+1)(\alpha-\bar{k}+2)}{2}$ coincides exactly with the number of parameters $\{e^{(k)}_{ij}\}$. Similar analysis applied to ${\tilde{G}_{k+1}}$ shows that equations (\ref{qo1}-\ref{qo3}) lead to the same constraints thanks to the symmetry between (\ref{qo2}) and (\ref{qo3}) under the operators' exchange ${\cW}_{-k}\leftrightarrow{\cW}_{k+1}$ and ${\cG}_{k+1}\leftrightarrow{\tilde{\cG}}_{k+1}$. Then, any higher element ${\cG}_{k+1}$ (and similarly for $\tilde{\cG}_{k+1}$) admits an expansion of the form (\ref{ansatz}), this combination solves (\ref{qo1}-\ref{qo2}) and is unique up to (\ref{qo1})-(\ref{qo11}) with $k\rightarrow i,l\rightarrow j$. \vspace{1mm}

It remains to show that (\ref{ansatz}) satisfies the remaining equations (\ref{qo6}-\ref{qo11}). In \cite{BSh1}, recall that the isomorphism between the reflection equation algebra associated with the $U_q(\widehat{sl_2})$ $R-$matrix and (\ref{qo1}-\ref{qo11}) was established. In particular, equations (\ref{qo1}-\ref{qo3}) and (\ref{qo4}) at $l=0$  are sufficient to determine uniquely the solution of the reflection equation algebra. Thanks to the isomorphism between the reflection equation algebra and $O_q(\widehat{sl_2})$, it follows that (\ref{qo6}-\ref{qo11}) are automatically satisfied if (\ref{sol}) are satisfied. This completes the proof of the proposal.

\vspace{0.5cm}
\noindent{\bf Acknowledgements:}  
S.B. thanks LMPT for hospitality, where part of this work has been done, and INFN Iniziativa Specifica FI11 for financial support.

\vspace{0.2cm}

\end{document}